\documentclass[american,aps,pra,reprint,floatfix,nofootinbib,superscriptaddress]{revtex4-1}
\usepackage[unicode=true,pdfusetitle, bookmarks=true,bookmarksnumbered=false,bookmarksopen=false, breaklinks=false,pdfborder={0 0 0},backref=false,colorlinks=false]{hyperref}
\hypersetup{colorlinks,linkcolor=red,citecolor=myurlcolor,urlcolor=myurlcolor}
\usepackage{graphics,epstopdf,graphicx,amsthm,amsmath,amssymb,mathptmx,braket,colortbl,color,bm,framed,mathrsfs}
\usepackage[T1]{fontenc}
\usepackage[up]{subfigure}
\usepackage{tikz}

\definecolor{myurlcolor}{rgb}{0,0,0.9}

\newcommand{\proj}[1]{| #1\rangle\!\langle #1 |}

\newcommand{\iinner}[2]{\langle #1 | #2\rangle}

\DeclareMathOperator{\trace}{Tr}
\newcommand{\Ptr}[2]{\trace_{#1}\Pa{#2}}
\newcommand{\Tr}[1]{\Ptr{}{#1}}
\newcommand{\Innerm}[3]{\left\langle #1 \left| #2 \right| #3 \right\rangle}
\newcommand{\Pa}[1]{\left[#1\right]}

\newcommand{\norm}[1]{\left\lVert #1 \right\rVert}

\theoremstyle{plain}
\newtheorem{thm}{Theorem}

\newtheorem{prop}[thm]{Proposition}

\newcommand*{\myproofname}{Proof}

\def\ot{\otimes}
\def\complex{\mathbb{C}}

\def\cI{\mathcal{I}}
\def\cA{\mathcal{A}}
\def\cG{\mathcal{G}}

\def\cD{\mathcal{D}}

\def\cH{\mathcal{H}}

\def\md{{(\mathrm{d})}}

\makeatother

\begin{document}

  \author{Kaifeng Bu}
 \email{bkf@zju.edn.cn}
 \affiliation{School of Mathematical Sciences, Zhejiang University, Hangzhou 310027, PR~China}
 \author{Namit Anand}
 \email{namitniser11@gmail.com}
 \affiliation{Department of Physics, National Institute of Science Education and Research, Bhubaneswar, 752050, India}
 \affiliation{Harish-Chandra Research Institute, Allahabad, 211019, India}
 \author{Uttam Singh}
 \email{uttamsingh@hri.res.in}
 \affiliation{Harish-Chandra Research Institute, Allahabad, 211019, India}
\affiliation{Homi Bhabha National Institute, Training School Complex, Anushakti Nagar, Mumbai 400085, India}

\title{Asymmetry and coherence weight of quantum states}

\begin{abstract}
The asymmetry of quantum states is an important resource in quantum information processing tasks such as quantum metrology and quantum communication. In this paper, we introduce the notion of \emph{asymmetry weight} --- an operationally motivated asymmetry quantifier in the resource theory of asymmetry. We study the convexity and monotonicity properties of asymmetry weight and focus on its interplay with the corresponding semidefinite programming (SDP) forms along with its connection to other asymmetry measures. Since the SDP form of asymmetry weight is closely related to asymmetry witnesses, we find that the asymmetry weight can be regarded as a (state-dependent) asymmetry witness. Moreover, some specific entanglement witnesses can be viewed as a special case of an asymmetry witness --- which indicates a potential connection between asymmetry and entanglement. We also provide an operationally meaningful coherence measure, which we term \emph{coherence weight}, and investigate its relationship to other coherence measures like the robustness of coherence and the $l_1$ norm of coherence. In particular, we show that for Werner states in any dimension $d$, all three coherence quantifiers, namely, the coherence weight, the robustness of coherence, and the $l_1$ norm of coherence, are equal and are given by a single letter formula.
\end{abstract}

\maketitle

\section{Introduction}
The role of symmetry in physics is essentially \textit{twofold} --- it provides both a \textit{constraint} on the dynamics and a \textit{simplification} in the structure of a theory. The special theory of relativity, for example, is a theory based on the constraints that the laws of physics remain invariant in all inertial frames of reference and that the speed of light in vacuum is the same for all observers, regardless of the state of motion of the light source. These constraints manifest themselves in the form of Lorentz invariance (or Poincar\'e invariance, more generally) of physical quantities and in turn, provide simplifications in the calculations in this theory. Similarly, in quantum theory, the presence of continuous symmetries like space and time translation invariance, and discrete symmetries like parity and time reversal, often helps in simplifying a given problem. At the same time, they can also manifest themselves as constraints in the form of \textit{superselection} rules --- postulated rules that forbid the preparation of quantum states that exhibit coherence between eigenstates of certain observables.\footnote{Note that superselection rules need not necessarily originate from an underlying symmetry.}\\

For ubiquitous physical systems, dynamics can be so complex that the only way to characterize their evolution is through the study of underlying symmetries --- which could, otherwise, be so tortuous that one could not possibly hope to study them. Symmetry, therefore, takes a pivotal stance in the fundamental process of deciphering the nature of the physical world. It is, then, not hard to see that if the dynamics of a physical system respects certain symmetries, then the states evolving under such processes cannot generate any more asymmetry than they already began with --- which motivates the popular dictum --- \textit{symmetric dynamics cannot generate asymmetry}. However, not all asymmetry is bad; given that local operations and classical communication cannot generate entanglement, entangled states can be seen as asymmetry-carrying states which, as is known, are extremely useful when it comes to quantum information processing.\\

In this paper, we would like to take the outlook of symmetries as constraints and use them to construct the corresponding resource theories. Once we have identified the presence of a symmetry in a given scenario, the constraints arise naturally. Every constraint on quantum operations, in turn, defines a resource theory ---determining how quantum states that cannot be prepared under the constraint may be used to outmaneuver the restriction. A resource theory is usually composed of two basic elements: the free states and the free operations. The set of allowed states (operations) under the given constraint is what we call the set of ``free'' states (operations). States of a physical system that do not satisfy the said symmetry are called \textit{asymmetric states} and may become useful as a resource for various tasks in the presence of the constraint. This is precisely the content of resource theories of asymmetry --- the quantification and manipulation of asymmetric states as a resource.\\

A mathematical entity is called symmetric if it is invariant under the action of a symmetry group ${\sf G}$. The resource theory of asymmetry is then defined with respect to a desired representation of a symmetry group and has led to a plethora of interesting results in the area of quantum information theory \cite{Bartlett2007,Gour2008,Gour2009,Marvian2012,Marvian2013,Marvian14,Marvian2014}. One of the earliest resource theories is that of quantum entanglement \cite{HorodeckiRMP09}, which is a basic resource for various quantum information processing protocols such as
superdense coding \cite{Bennett1992}, remote state preparation \cite{Pati2000,Bennett2001} and quantum teleportation \cite{Bennett1993}. Other notable examples include the resource theories of thermodynamics \cite{Fernando2013}, coherence \cite{Baumgratz2014,Girolami2014,Streltsov2015,Winter2016,Killoran2016,Chitambar2016,Chitambar2016a}, superposition \cite{Theurer2017}, and steering \cite{Rodrigo2015}. Interestingly, the notion of resource theories can be generalized to include any description or knowledge that one may have of a physical state \cite{Rio2015}. The asymmetry of states is germane to quantum information theory and has interesting applications ranging from quantum metrology \cite{Giovannetti2004, Giovannetti2006, Giovannetti2011} to quantum communication \cite{Duan2011,Gisin2007}. Interesting experimental progress has been reported in this direction recently; for example, metrologically useful asymmetry and entanglement were detected in an all-optical experiment by studying how these resources affect the speed of evolution of a quantum system under a unitary transformation \cite{Zhang2017}. \\

One of the main advantages that a resource theory offers is the lucid quantitative and operational description as well as the manipulation of the relevant resources at one's disposal. The robustness-based quantifiers capture the robustness of a given resource to noise and form an operationally powerful method to quantify the resource. The quantifiers that are obtained from this method include the robustness of entanglement \cite{Vidal1999},  the robustness of steering \cite{Piani2015PRL}, the robustness of asymmetry \cite{Piani2016}, and the robustness of coherence \cite{Napoli2016}. Another important class of quantifiers which are known as the resource weight-based quantifiers is defined by the smallest amount of resource needed to prepare a given state. The best separable approximation of entangled states \cite{Lewenstein1998}, the steering weight \cite{Skrzypczyk2014}, and the measurement incompatibility weight \cite{Pusey2015} are some examples of the same. Fortunately, both the robustness-based and the weight-based quantifiers are easy to calculate numerically since they can be characterized as the solutions to the corresponding semidefinite programming (SDP) forms \cite{Piani2016, semidefinite1996}. \\

In this paper, we introduce weight-based quantifiers in the resource theories of asymmetry and coherence and term them as the asymmetry and coherence weight, respectively. We then prove several properties such as their convexity and monotonicity under free operations and also provide the corresponding SDP forms which make the numerical calculations tractable. The SDP form indicates that the asymmetry weight can be regarded as a state-dependent asymmetry witness and some entanglement witnesses may be viewed as a special asymmetry witness (see also, Ref. \cite{Girolami2015A} for a discussion on the connection between entanglement and asymmetry witnesses). Additionally, we find interesting relationships between coherence (asymmetry) weight and other coherence (asymmetry) measures. For pure coherent (asymmetric) states, the coherence (asymmetry) weight is always equal to 1 --- suggesting the coarse-grained nature of the coherence (asymmetry) weight. Then, we consider a broad class of mixed bipartite quantum states, namely, the generalized $X$ states, and find analytical expressions for their coherence weight. Moreover, for Werner states, we show that the coherence weight, the robustness of coherence, and the $l_1$ norm of coherence are all equal. Furthermore, in the context of distribution of coherence, we provide some useful inequalities between the coherence weight (robustness of coherence) of bipartite quantum states and that of their marginals.\\

The paper is organized as follows. We start by giving an exposition to the resource theories of asymmetry and coherence and a brief overview of semidefinite programming in Sec. \ref{sec:pre}. In Sec. \ref{sec:asym}, we give the definition of asymmetry weight and investigate the convexity and monotonicity properties and show how to cast it in the form of a semidefinite program. Additionally, we find various relationships between the asymmetry weight and other asymmetry measures. Moreover, we give the definition of coherence weight and investigate the properties of coherence weight in a similar spirit to the asymmetry weight in Sec. \ref{sec:coh}, and in particular, we obtain explicit analytical results for Werner states, Gisin states, and in general, the generalized $X$ states (Sec. \ref{subsec:exact_coh}). Finally, we conclude with a brief discussion and overview of the results obtained in this paper in Sec. \ref{sec:con}.

\section{Preliminaries}\label{sec:pre}

\subsection{Resource theory of asymmetry}
The resource theory of asymmetry with respect to a given representation of a symmetry group ${\sf G}$ has been used extensively to distinguish and quantify the symmetry-breaking properties of both the states and the operations \cite{Gour2008,Marvian2012,Marvian2013}. Given a Hilbert space $\cH$ and the convex set ${\mathscr{D}}(\cH)$ of density operators acting on it, let us consider a symmetry group ${\sf G}$ with an associated unitary representation  $\{U_g\}_{g \in {\sf G}}$ on $\cH$. The free states in the resource theory of asymmetry are called \textit{symmetric} states \cite{Gour2008,Marvian2012,Marvian2013} and the set of symmetric states
in $\cD(\cH)$ is defined as
\begin{eqnarray*}
\mathcal{J}=\set{\sigma\in\cD(\cH):\mathbf{U}_g(\sigma) =\sigma, \forall g\in \sf G },
\end{eqnarray*}
where $\mathbf{U}_g(\sigma):=U_g\sigma U^\dag_g$ \cite{Gour2008,Marvian2012,Marvian2013,Piani2016}.
The set $\mathcal{J}$ can also be written as
\begin{eqnarray}
\mathcal{J}=\set{\sigma\in\cD(\cH):\mathcal{G}(\sigma)=\sigma},
\end{eqnarray}
where $\mathcal{G}(\sigma)=\frac{1}{|\sf G|}\sum_g \mathbf{U}_g(\sigma)$ is the group average \cite{Gour2008,Marvian2012,Marvian2013,Piani2016}. The free operations that we consider in this paper are the selective covariant operations with respect to $\sf G$ \cite{Marvian2012}. For any such quantum operation $\Phi=\sum_i\Phi_i$, then, $[\Phi_i, \mathbf{U}_g]=0$ for any $i$, $\forall g\in \sf G$ \cite{Marvian2012}. For example, the Kraus representation of a quantum operation $\Phi$ can be written in the above form by considering the suboperation $\Phi_i$ as: $\Phi_i (\rho) = K_i \rho K^{\dag}_i$.

Given a quantum state $\rho\in \cD(\cH)$, the relative entropy of asymmetry $\cA_r$ \cite{Marvian2014} and the robustness of asymmetry $\cA_R$ \cite{Piani2016} are defined, respectively, as
\begin{eqnarray}
\cA_r(\rho)&=&S(\rho||\mathcal{G}(\rho))=S(\mathcal{G}(\rho))-S(\rho),\\
\cA_R(\rho)&=&\min_{\tau}\left\{ s\geq0: \frac{\rho+s\tau}{1+s}\in \mathcal{J}, \tau \in \cD(\cH) \right\}.
\end{eqnarray}
Here $S$ is the von Neumann entropy and, for a state $\rho$, $S(\rho) = -\mathrm{Tr}[\rho \ln \rho]$. \\ 

\subsection{Resource theories of quantum coherence}
Given a fixed reference basis, say \{$|i \rangle_{i=0,\ldots,d-1}$\} for some $d$-dimensional Hilbert space, $\mathcal{H}^d$, any quantum state which is diagonal in the reference basis is called an incoherent state and is a free state in the resource theory of coherence. However, note that there is still no general consensus on the set of allowed operations in the resource theory of coherence and, for example, we have resource theories of coherence based on incoherent operations and symmetric operations \cite{Gour2008,Marvian14,Baumgratz2014,Chitambar2016a}. More details can be found in a recent review on the resource theories of coherence \cite{streltsov2016quantum}. \\

In this paper, we consider the resource theory of coherence based on incoherent operations \cite{Baumgratz2014}. Let $\mathcal{I}$ be the set of all incoherent states. An operation $\Phi$ is called an incoherent operation if the set of Kraus operators $\{K_i\}$ of $\Phi$ is such that $K_i\mathcal{I} K^{\dag}_i\subseteq \mathcal{I}$ for each $i$. For a $d$-dimensional quantum system in a state $\rho$ and a fixed reference basis $\{\ket{i}\}$, the $l_1$ norm of coherence $C_{l_1}(\rho)$ \cite{Baumgratz2014}, the relative entropy of coherence $C_r(\rho)$ \cite{Baumgratz2014},
 and the robustness of coherence $C_R(\rho)$ \cite{Napoli2016} are defined, respectively, as
\begin{eqnarray}
C_{l_1}(\rho) &=& \sum_{\substack{{i,j=0}\\{i\neq j}}}^{d-1} |\bra{i}\rho\ket{j}|;\\
C_r(\rho)&=&S(\rho^{\md})-S(\rho),~ \rho^{\md}=\sum_k\rho_{kk}\proj{k};\\
C_R(\rho)&=&\min_{\tau}\left\{s\geq0: \frac{\rho+s\tau}{1+s}\in \mathcal{I}, \tau \in \cD(\cH) \right\}.
\end{eqnarray}

\subsection{Semidefinite programming}
SDP is a powerful tool in combinatorial optimization, which is a generalization of linear programming problems \cite{semidefinite1996}. A semidefinite program over $\mathcal{X}=\mathbb{C}^N$ and $\mathcal{Y}=\mathbb{C}^M$ is defined as a triple ($\Psi$, $C$, $D$), where $\Psi$ is a Hermiticity-preserving map from $\mathcal{L(X)}$ (linear operators on $\mathcal{X}$) to $\mathcal{L(Y)}$ (linear operators on $\mathcal{Y}$), $C\in$ Herm($\mathcal{X}$) (Hermitian operators over $\mathcal{X}$), and $D\in$ Herm($\mathcal{Y})$ (Hermitian operators over $\mathcal{Y}$).
There is a pair of optimization problems associated with every semidefinite program ($\Psi$, $C$, $D$), known as the primal and the dual problems. The standard form of a semidefinite program (that is typically used for general conic programming) is \cite{watrous2009}
\begin{equation}
\begin{matrix}
\textrm{\underline{Primal problem}} \textrm{                          } &  \textrm{\underline{Dual problem}} \vspace{2mm}\\ 
\textrm{maximize: } \langle C,X \rangle, \textrm{                      } & \textrm{minimize: } \langle D,Y \rangle, \\
\textrm{subject to: } \Psi(X) \leq D, \textrm{                      } & \textrm{subject to: } \Psi^*(Y) \geq C, \\
X \in Pos(\mathcal{X}). & Y \in Pos(\mathcal{Y}).
\end{matrix}
\end{equation}
SDP forms have interesting and ubiquitous applications in quantum information theory \cite{watrous2009}; for example, it was recently shown by Brandao and Svore \cite{brandao2016quantum} that there exists a quantum algorithm for solving semidefinite programs that gives an unconditional square-root speedup over any existing classical method.\\

\section{Asymmetry weight}\label{sec:asym}
The weight-based quantifier in the resource theory of asymmetry which we call the \textit{asymmetry weight} is given in an operationally motivated way and will be proved to satisfy all the conditions that a proper asymmetry measure needs to fulfill. We also give the corresponding SDP form and show how the asymmetry weight can be viewed as a state-dependent asymmetry witness.\\

\begin{figure}
\includegraphics[scale=1]{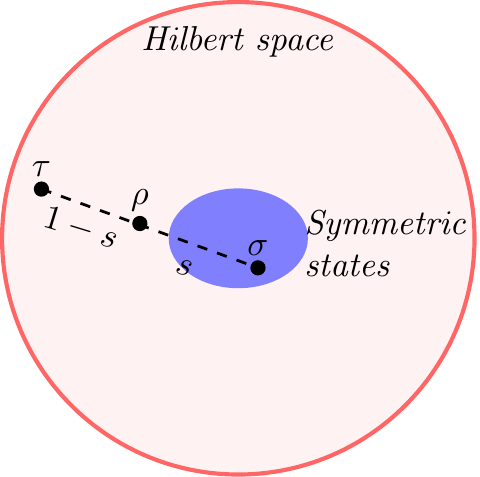}
\caption{(Color online) The set of symmetric states $\mathcal{J}$ (shown in blue) forms a subspace of the total Hilbert space $\mathcal{H}$ (shown in red). The asymmetry weight of a quantum state $\rho$ is then defined as the minimum weight convex mixture of $\sigma$ and $\tau$ where $\sigma \in \mathcal{J}$ and $\tau \in \cD(\cH)$.}
\label{fig:schematic}
\end{figure}

\noindent
{\bf Definition 1.} In the process of preparing some given quantum state $\rho$, we want to use the least number of asymmetry resources --- which means that we would like to use the symmetric states as much as possible and the asymmetric states as little as possible --- such that we generate the given state $\rho$ on an average. That is, given a quantum state $\rho$, the asymmetry weight of $\rho$ is defined as (see also Figure \ref{fig:schematic})
\begin{eqnarray}\label{def1:asy_w}
 \cA_w(\rho)=\min_{\{\sigma,\tau\}}\left\{s\geq0:\rho=(1-s)\sigma+s\tau,  \sigma\in \mathcal{J}, \right.\nonumber\\
 \left.\tau\in\cD(\cH)\right\}.
\end{eqnarray}
The asymmetry weight, defined as above, has some nice properties such as convexity and monotonicity under covariant operations, which we will prove in the following.

\begin{prop}
Given a quantum state $\rho\in\cD(\cH)$, the asymmetry weight $\cA_w(\rho)$ is bounded as
$0\leq \cA_w(\rho)\leq 1$, and $\cA_w(\rho)=0$ iff $\rho\in \mathcal{J}$, i.e., iff $\rho$ is symmetric.
\end{prop}

\begin{proof}
The proof follows directly from the definition of asymmetry weight.
\end{proof}

\begin{prop}
The asymmetry weight $\cA_w$ is convex in $\rho$, i.e.,
\begin{eqnarray*}
\cA_w(p\rho_1+(1-p)\rho_2)\leq p \cA_w(\rho_1)+(1-p)\cA_w(\rho_2),
\end{eqnarray*}
where $p\in [0,1]$ and $\rho_1,\rho_2\in \cD(\cH)$.
\end{prop}
\begin{proof}
Let $\rho_1=(1-\cA_w(\rho_1))\sigma^*_1+\cA_w(\rho_1)\tau^*_1$ and $\rho_2=(1-\cA_w(\rho_2))\sigma^*_2+\cA_w(\rho_2)\tau^*_2$ be the optimal decompositions for $\rho_1$ and $\rho_2$, where $\sigma^*_1,\sigma^*_2\in\mathcal{J}$ and $\tau^*_1, \tau^*_2\in\cD(\cH)$.
Let
\begin{eqnarray*}
\sigma&=&\frac{1}{1-s}[p(1-\cA_w(\rho_1))\sigma^*_1+(1-p)(1-\cA_w(\rho_2))\sigma^*_2],\\
\tau&=&\frac{1}{s}[p\cA_w(\rho_1)\sigma^*_1+(1-p)\cA_w(\rho_2)\sigma^*_2],\\
s&=&p\cA_w(\rho_1)+(1-p)\cA_w(\rho_2).
\end{eqnarray*}
Then
\begin{eqnarray}
p\rho_1+(1-p)\rho_2=(1-s)\sigma+s\tau,
\end{eqnarray}
which implies that $\cA_w(p\rho_1+(1-p)\rho_2)\leq s=p\cA_w(\rho_1)+(1-p)\cA_w(\rho_2) $.
\end{proof}

\begin{prop}
Let $\Phi=\sum_i\Phi_i$ be a selective $\sf G$-covariant  quantum operation; i.e.,
for any $i$, $[\Phi_i, \mathbf{U}_g]=0,\forall g\in \sf G$. Then, the asymmetry weight is monotonically nonincreasing on an average:
\begin{eqnarray}
\cA_w(\rho)\geq \sum_i p_i \cA_w(\rho_i),
\end{eqnarray}
where $p_i:=\mathrm{Tr}[\Phi_i(\rho)]$ and $\rho_i=\frac{\Phi_i(\rho)}{p_i}$.
\end{prop}
\begin{proof}
Let $\rho=(1-\cA_w(\rho))\sigma^*+\cA_w(\rho)\tau^*$ be the optimal decomposition, where $\sigma^*\in\mathcal{J}$ and $\tau^*\in\cD(\cH)$. Then
\begin{eqnarray*}
\Phi_i(\rho)=(1-\cA_w(\rho))\Phi_i(\sigma^*)+
\cA_w(\rho)\Phi(\tau^*).
\end{eqnarray*}
Let
\begin{eqnarray*}
\sigma_i&=&\frac{1}{(1-s_i)p_{i}}(1-\cA_w(\rho))\Phi_i(\sigma^*),\\
\tau_i&=&\frac{1}{s_i p_{i}}\cA_w(\rho)\Phi_i(\tau^*),\\
s_i&=&\frac{1}{p_i}\cA_w(\rho)\Tr{\Phi_i(\tau^*)},
\end{eqnarray*}
then $\rho_i=(1-s_i)\sigma_i+s_i\tau_i$. As $\Phi_i(\sigma^*)\in \mathcal{J}$, then $\cA_w(\rho_i)\leq s_i$. Therefore, $\sum_i p_i \cA_w(\hat{\rho}_i) \leq\sum_i p_i s_i =\cA_w(\rho)$. This concludes the proof of the proposition.
\end{proof}
Using the above three Propositions, we have shown that the asymmetry weight is a proper asymmetry measure. In the following, we express the asymmetry weight in terms of semidefinite programs and explore its connection to asymmetry and entanglement witnesses.\\

\subsection{Asymmetry weight as a semidefinite program}
A decomposition of a given state $\rho=(1-s)\sigma+s\tau$ is equivalent to the condition
$\rho\geq (1-s)\sigma$, where $\sigma \in\mathcal{J}$, as there exists a quantum state $\tau \in \cD(\cH)$ such that $\rho-(1-s)\sigma=s\tau$ if $\rho\geq (1-s)\sigma$. Then, the asymmetry weight can also be characterized as
 \begin{eqnarray}\label{def2:asy_w}
  \cA_w(\rho)=\min_{\sigma\in \mathcal{J}}\Set{s\geq0:\rho\geq(1-s)\sigma}.
 \end{eqnarray}
In view of the formula \eqref{def2:asy_w}, the SDP form of asymmetry can be obtained as follows.

\begin{thm}\label{thm:A_w}
For a given state $\rho\in\cD(\cH)$, the asymmetry weight $\cA_w(\rho)$ can be characterized as the solution of the following optimization problem:

\begin{eqnarray}\label{eq:AW_f}
\max \Tr{\rho W},\\\nonumber
\text{such that}~~ \mathcal{G}(W)\leq 0,\\\nonumber
W\leq \mathbb{I},\nonumber
\end{eqnarray}
where operator  $W$ is Hermitian.
\end{thm}

\begin{proof}
Since $\cA_w(\rho)$ can be written as
\begin{eqnarray*}
\cA_w(\rho)=\min_{\sigma\in \mathcal{J}}\set{s\geq 0:\rho\geq (1-s)\sigma},
\end{eqnarray*}
$1-\cA_w(\rho)$ can be obtained as the solution of the following semidefinite program:
\begin{eqnarray}\label{eq:SDP1}
\max \Tr{\widetilde{\sigma}},\\\nonumber
\text{such that}~~\widetilde{\sigma}\leq \rho,\\\nonumber
\mathcal{G}(\widetilde{\sigma})=\widetilde{\sigma};\\\nonumber
\widetilde{\sigma}\geq 0,
\end{eqnarray}
where $\widetilde{\sigma} = (1-s)\sigma$. Let us consider the following semidefinite program:
\begin{eqnarray}\label{eq:SDP2}
\max \Tr{\widetilde{\sigma}},\\\nonumber
\text{such that}~~\mathcal{G}(\widetilde{\sigma})\leq \rho,\\\nonumber
\widetilde{\sigma}\geq 0.
\end{eqnarray}
Since $\widetilde{\sigma}\leq \rho$ and $\mathcal{G}(\widetilde{\sigma})=\widetilde{\sigma}$,
then $\mathcal{G}(\widetilde{\sigma})\leq \rho$, which implies that the solution of \eqref{eq:SDP2} is an upper bound of \eqref{eq:SDP1}. Additionally, as $\mathcal{G}^2(\widetilde{\sigma})=\mathcal{G}(\widetilde{\sigma})$ and $\Tr{\mathcal{G}(\widetilde{\sigma})}=\Tr{\widetilde{\sigma}}$, we have that
\eqref{eq:SDP1} is an upper bound of \eqref{eq:SDP2}. Thus, \eqref{eq:SDP1} and \eqref{eq:SDP2}  are equivalent, that is, $\mathcal{A}_w(\rho)$ can be characterized as the solution of  \eqref{eq:SDP2}.

The standard form of  \eqref{eq:SDP2}  is
\begin{eqnarray*}
\max \Tr{B\widetilde{\sigma}},\\
\text{such that}~~ \Lambda^\dag(\widetilde{\sigma})\leq C,\\
\widetilde{\sigma}\geq 0,
\end{eqnarray*}
with $B=\mathbb{I}, \Lambda^\dag(\widetilde{\sigma})=\mathcal{G}(\widetilde{\sigma})$, and $
C=\rho$. The dual semidefinite program is
\begin{eqnarray*}
\min \Tr{CX},\\
\text{such that}~~ \Lambda(X)\geq B,\\
X\geq 0.
\end{eqnarray*}
That is,
\begin{eqnarray*}
\min \Tr{\rho X},\\
\text{such that}~~ \mathcal{G}(X)\geq \mathbb{I},\\
X\geq 0.
\end{eqnarray*}
Note that the dual is strictly feasible as we only need to choose $X=\alpha \mathbb{I}$ for a large enough $\alpha$. Thus, strong duality holds, which implies that $\cA_w$ can be viewed as the solution of the
following semidefinite program:
\begin{eqnarray*}
\max 1-\Tr{\rho X},\\
\text{such that}~~ \mathcal{G}(X)\geq \mathbb{I},\\
X\geq 0.
\end{eqnarray*}
Take $W=\mathbb{I}-X$, then
\begin{eqnarray*}
\max \Tr{\rho W},\\
\text{such that}~~ \mathcal{G}(W)\leq 0,\\
W\leq\mathbb{ I}.
\end{eqnarray*}
This concludes the proof of the theorem.
\end{proof}

\subsection{Asymmetry witness as an entanglement witness}
For any Hermitian operator $W$ with $\mathcal{G}(W)\leq 0$ and for any state $\sigma\in \mathcal{J}$, we have $\Tr{\sigma W}\leq 0$. Any such $W$ can be viewed as an asymmetry witness (similar to the idea of asymmetry witness defined in Ref. \cite{Piani2016} up to a signal), as for any state $\rho \in \cD(\cH)$ with $\Tr{\rho W}>0$ implies that  $\rho$ is asymmetric. In quantum entanglement theory, entanglement witnesses have been introduced to detect entanglement \cite{Horodecki1996,Terhal2000}. For example, the swap entanglement witness $V=\sum_{ij}\ket{ij}\bra{ji}$ has been used to indicate the existence of  entanglement where $\Tr{\rho V}<0$ \cite{Werner1989,Horodecki1996}. Since the group $\set{\mathbb{I}, V}$ is a unitary representation of the symmetry group $S_2$ on the the Hilbert space $\complex^d\ot\complex^d$, let us take $W=-V$, which implies that $\cG(W)=-2V\leq 0$. That is, the swap entanglement witness may be regarded as a special asymmetry witness. This goes on to accentuate the interplay between nonclassicality like coherence and squeezing and
quantum correlations like discord \cite{discord2000zurek,discord2001zurek,discord2001vedral} and entanglement \cite{Killoran2016}. A recent linear optics experiment took an important step in this direction where coherence in a local system was consumed to synthesize an identical amount of quantum discord with an ancilla system using only incoherent operations \cite{Wu2017}.\\

Additionally, given a quantum state $\rho$, there exists the optimal choice $W^*_\rho$ which depends on $\rho$, such that $\cA_w=\Tr{\rho W^*_\rho}$. That is, $\cA_w$ can be viewed as a quantitative asymmetry witness (which is state dependent). In fact, we find that many other asymmetry measures can be regarded as quantitative asymmetry witnesses, for example, the relative entropy of asymmetry $\cA_r(\rho)=S(\rho||\mathcal{G}(\rho))= \Tr{\rho W^s_\rho}$, where $W^s_\rho=\log\rho-\log \mathcal{G}(\rho)$ is an asymmetry witness as $\mathcal{G}(W^s_\rho)\leq 0$. Also, the robustness of asymmetry can be expressed as $\cA_R(\rho)=-\Tr{\rho W^r_\rho}$, where $\mathcal{G}(W^r_\rho)=0$ (see Ref. \cite{Piani2016}).\\

Furthermore, in view of the fact that the asymmetry weight can be expressed as an asymmetry witness, we can get a lower bound on the asymmetry weight by using the Hilbert-Schmidt distance between $\rho$ and $\cG(\rho)$.

\begin{prop}
\label{prop-asym-weight-dis}
For any given  state $\rho\in \cD(\cH)$, we have
\begin{eqnarray}
\cA_w(\rho)\geq
\frac{\norm{\rho-\mathcal{G}(\rho)}^2_2}{\norm{\rho}_\infty}
\geq \norm{\rho-\mathcal{G}(\rho)}^2_2,
\end{eqnarray}
where $||A||_{2}^{2}:=\Tr{A^{\dagger}A}$ is the Hilbert-Schmidt norm and $||A||_{\infty}:=\max_{i} \lambda_i$ with $\lambda_{i}$ being the $i$th eigenvalue of $\sqrt{A^{\dagger}A}$.
\end{prop}
\begin{proof}
Let $W=\frac{\rho-\mathcal{G}(\rho)}{\norm{\rho}_\infty}$, then $\mathcal{G}(W)=0$ and $
W\leq\frac{\rho}{\norm{\rho}_\infty}\leq \mathbb{I}$.
Thus,
\begin{eqnarray*}
\cA_w(\rho)&\geq& \Tr{\rho W}\\
&=& \frac{\Tr{\rho(\rho-\mathcal{G}(\rho))}}{\norm{\rho}_\infty}\\
&=& \frac{\Tr{\rho^2}-\Tr{\mathcal{G}(\rho)^2}}{\norm{\rho}_\infty}\\
&=&\frac{\norm{\rho-\mathcal{G}(\rho)}^2_2}{\norm{\rho}_\infty}\\
&\geq&\norm{\rho-\mathcal{G}(\rho)}^2_2.
\end{eqnarray*}
Here, we use the fact that
$\Tr{\rho\mathcal{G}(\rho)}=\Tr{\mathcal{G}(\rho)^2}$.
\end{proof}

The distance between the  state $\rho$ and $\cG(\rho)$ can be used to quantify how asymmetric the state $\rho$ is, as the state $\rho$ is symmetric iff $\rho=\cG(\rho)$. The above proposition facilitates a connection between the asymmetry weight and the Hilbert-Schmidt distance. In the following, we find the relationship between the asymmetry weight and other asymmetry measures.

\begin{prop}\label{prop:Aw_vs_A_R}
Given a quantum state $\rho\in\cD(\cH)$, we have
\begin{eqnarray}
\cA_w(\rho)&\geq& \frac{1}{d-1}\cA_R(\rho),~ \mathrm{and}\\
\cA_w(\rho)&\geq& \frac{1}{\ln d}\cA_r(\rho).
\end{eqnarray}

\end{prop}
\begin{proof}
Let $\rho=(1-\cA_w(\rho))\sigma^*+\cA_w(\rho)\tau^*$ be the optimal decomposition, where $\sigma^* \in \mathcal{J}$. Since  $\cA_R$ is convex \cite{Piani2016},
\begin{eqnarray*}
\cA_R(\rho)&\leq&(1-\cA_w(\rho))\cA_R(\sigma^*)+\cA_w(\rho)\cA_R(\tau^*)\\
&=&\cA_w(\rho)\cA_R(\tau^*),
\end{eqnarray*}
where $\cA_R(\sigma^*)=0$ comes from the fact that  $\sigma^*\in \mathcal{J}$. Additionally, $\cA_R(\tau^*)\leq d-1$  implies that
\begin{eqnarray*}
\cA_w(\rho)\geq \frac{\cA_R(\rho)}{\cA_R(\tau^*)}\geq \frac{1}{d-1}\cA_R(\rho).
\end{eqnarray*}
Similarly, from the convexity of the asymmetry measure $\cA_r(\rho)$ and using $\cA_r(\tau^*)\leq\ln d$, we get
\begin{eqnarray*}
\cA_w(\rho)&\geq& \frac{\cA_r(\rho)}{\cA_r(\tau^*)}\geq \frac{1}{\ln d}\cA_r(\rho).
\end{eqnarray*}
\end{proof}
\subsection{All pure asymmetric states have asymmetry weight 1}
If $\ket{\psi}$ is a pure asymmetric state, then its decomposition as $\proj{\psi}=(1-s)\sigma+s\tau$, where $\sigma\in \mathcal{J}$ and $\tau\in\cD(\cH)$, implies $\cA_w(\proj{\psi})=1$ as $s=1$. That is, for any pure asymmetric state, the asymmetry weight is always 1.\\

\section{Coherence weight}\label{sec:coh}
In this section, we introduce the coherence weight of a quantum state as a quantifier of coherence. Since coherence of a $d$-dimensional quantum system can be regarded as asymmetry with respect to a specific $d$-dimensional representation of $\sf G\equiv\sf U(1)$ \cite{Piani2016},  we can define the coherence weight of a given quantum state $\rho$ in a similar spirit as the asymmetry weight. That is,
 \begin{eqnarray}
 C_w(\rho)=\min_{\{\sigma,\tau\}}\left\{s\geq0:\rho=(1-s)\sigma+s\tau,  \sigma\in \mathcal{I}, \right.\nonumber\\
 \left.\tau\in\cD(\cH)\right\}.
 \end{eqnarray}
From the above definition, it is clear that the coherence weight, $C_w(\rho)$, of a given state $\rho$ can be interpreted operationally as the minimal number of genuine resource (coherent) states needed in the preparation process of the quantum state. According to \eqref{def2:asy_w}, coherence weight can also be expressed as
\begin{eqnarray}
 C_w(\rho)=\min_{\sigma\in \mathcal{I}}\set{s\geq0:\rho\geq(1-s)\sigma}.
\end{eqnarray}
In the following, by the incoherent operations we mean any quantum operation $\Phi$ with the Kraus representation $\set{K_i}$ such that $K_i\cI K^\dag_i\subseteq \cI$ for each $i$ \cite{Baumgratz2014}.

 \begin{prop}
Let $\Phi(\cdot)=\sum_i K_i\cdot K^\dag_i$ be an incoherent  operation with $K_i\cI K^\dag_i\subseteq \cI$ for each $i$. Then, the coherence weight of any state $\rho$ is monotonically nonincreasing on average, i.e.,
\begin{eqnarray}
C_w(\rho)\geq \sum_i p_i C_w(\rho_i),
\end{eqnarray}
where $p_i:=\Tr{K_i \rho K^\dag_i}$ and $\rho_i=\frac{K_i\rho K^\dag_i}{p_i}$.
\end{prop}

\begin{widetext}

\begin{figure}
\subfigure[]{\includegraphics[scale=0.25]{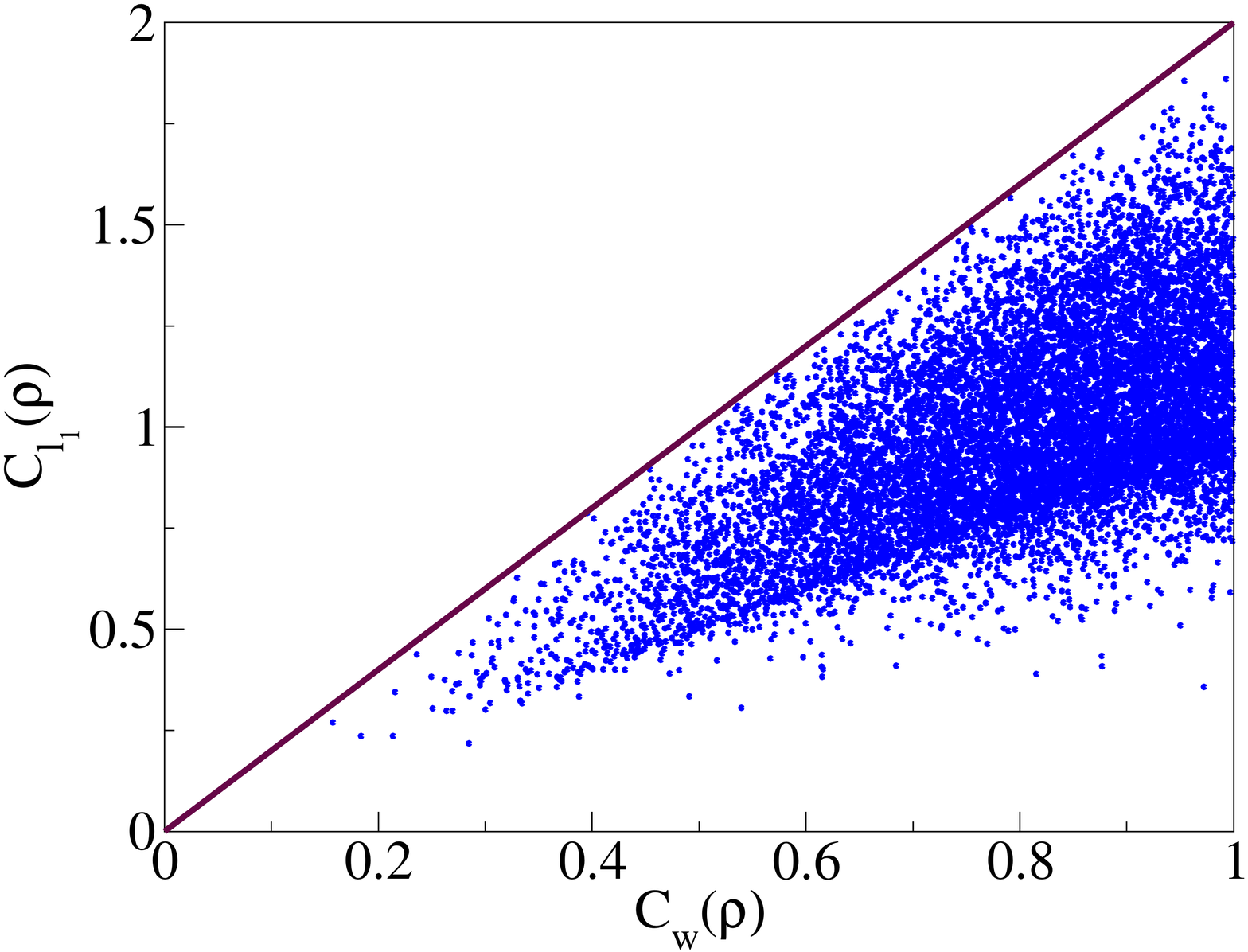}}
\subfigure[]{\includegraphics[scale=0.25]{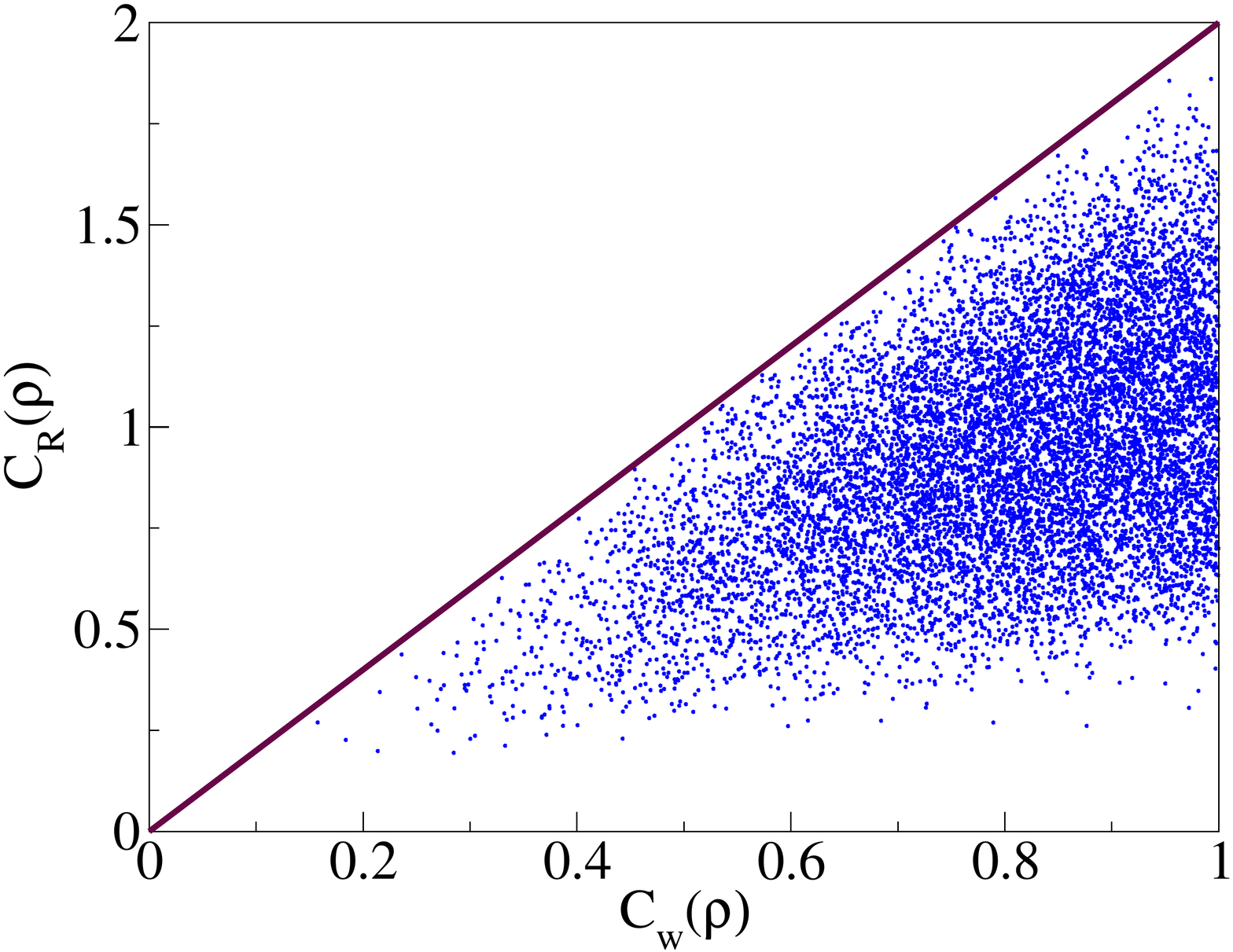}}
\caption{(Color online) Comparison between the coherence weight $C_w$ (horizontal axis) and (a) the $l_1$ norm of coherence $C_{l_1}$ and (b) the robustness of coherence  $C_{R}$ (vertical axes), for ${10}^4$ randomly generated quantum states (uniformly distributed with respect to the Haar measure) of dimension $3$. All the axes are unitless.}
\label{fig:cwvscl1vscrdim3}
\end{figure}

\end{widetext}

 \begin{proof}
Let $\rho=(1-C_w(\rho))\sigma^*+C_w(\rho)\tau^*$ be the optimal decomposition, where $\sigma^*\in\cI$ and $\tau^*\in\cD(\cH)$. Then
\begin{eqnarray*}
K_i\rho K^\dag_i=(1-C_w(\rho))K_i \sigma^* K^\dag_i+ C_w(\rho)K_i\tau^*K^\dag_i.
\end{eqnarray*}
Now, let us consider
\begin{eqnarray*}
\sigma_i&=&\frac{1}{(1-s_i)p_i}(1-C_w(\rho))K_i \sigma^* K^\dag_i,\\
\tau_i&=&\frac{1}{s_i p_i}C_w(\rho)K_i\tau^*K^\dag_i,\\
s_i&=&\frac{1}{p_i}C_w(\rho)\Tr{K_i\tau^*K^\dag_i}.
\end{eqnarray*}
Then, $\rho_i=(1-s_i)\sigma_i+s_i \tau_i$. As $\Phi_i(\sigma^*)\in \mathcal{J}$, we have $C_w(\rho_i)\leq s_i$. Therefore, $\sum_i p_i C_w(\rho_i) \leq\sum_i p_i s_i = C_w(\rho)$. This completes the proof of the proposition.
\end{proof}
Moreover, $C_w$ is a convex function of density matrices, i.e.,
\begin{eqnarray}
C_w(p\rho_1+(1-p)\rho_2)\leq p C_w(\rho_1)+(1-p)C_w(\rho_2),
\end{eqnarray}
where $p\in [0,1]$ and $\rho_1,\rho_2\in \cD(\cH)$. Since the definition of the coherence weight is very similar to the asymmetry weight, it can also be expressed as a semidefinite program.
\begin{prop}
The coherence weight $C_w$ can be characterized as the solution of the following semidefinite program:
\begin{eqnarray}\label{eq:CW_f}
\max \Tr{\rho W }, \nonumber \\
 \text{such that}~~\Delta(W)\leq0,\nonumber \\
W\leq \mathbb{I},\nonumber
\end{eqnarray}
where $\Delta(\cdot)=\sum_i\proj{i}\cdot\proj{i}$.
\end{prop}
\begin{proof}
The proof is similar to the proof of Theorem \ref{thm:A_w}.
\end{proof}
In the Supplemental Material \cite{supplement}, we provide the MATLAB \cite{matlab} code to evaluate the above semidefinite program numerically and calculate the coherence weight for a given state $\rho$ using the open-source MATLAB-based modeling system for convex optimization, CVX \cite{cvx1,cvx2}.  

\begin{prop}
\label{coh-weight-dis}
For a given  state $\rho\in \cD(\cH)$, we have
\begin{eqnarray}
C_w(\rho)\geq
\frac{\norm{\rho-\Delta(\rho)}^2_2}{\norm{\rho}_\infty}
\geq \norm{\rho-\Delta(\rho)}^2_2,
\end{eqnarray}
where $||A||_{2}^{2}:=\Tr{A^{\dagger}A}$ is the Hilbert-Schmidt norm and $||A||_{\infty}:=\max_{i} \lambda_i$ with $\lambda_{i}$ being the $i$th eigenvalue of $\sqrt{A^{\dagger}A}$.
\end{prop}

\begin{prop}\label{prop:Cw_vs_C_R}
Given a quantum state $\rho\in\cD(\cH)$, we have
\begin{eqnarray}
\label{ineq:CwvsCR}C_w(\rho)&\geq& \frac{1}{d-1}C_R(\rho),\\
\label{ineq:CwvsC1}C_w(\rho)&\geq& \frac{1}{d-1}C_{l_1}(\rho),\\
C_w(\rho)&\geq& \frac{1}{\ln d}C_r(\rho).
\end{eqnarray}
\end{prop}
The Propositions \ref{coh-weight-dis} and \ref{prop:Cw_vs_C_R} can be proved in a similar spirit as we have proved Propositions \ref{prop-asym-weight-dis} and \ref{prop:Aw_vs_A_R}, respectively. See also Figures \ref{fig:cwvscl1vscrdim3} and \ref{fig:cwvscl1vscrdim4}.

The $l_1$ norm of coherence has played a pivotal role in the quantification of coherence and its operational meaning has been investigated recently in Ref. \cite{Rana2016}. In the following, we explore the relationships between the coherence weight, the robustness of coherence, and the $l_1$ norm coherence. In addition to the simple connections \eqref{ineq:CwvsCR} and \eqref{ineq:CwvsC1}, we find better relationships between the three measures for special classes of states in finite-dimensional Hilbert spaces.

\begin{prop}\label{prop:6}
For a quantum state $\rho\in \cD(\cH)$, if there exists an incoherent unitary $U=\sum_j e^{i\phi_k}\proj{k}$ such that $\rho'=U\rho U^\dag$ with $\rho'_{ij}=-|\rho_{ij}|$, for any $i\neq j$, then
\begin{eqnarray}
C_w(\rho)\geq C_{l_1}(\rho).
\end{eqnarray}
\end{prop}
\begin{proof}
Similar to the method in \cite{Piani2016}, consider the optimal incoherent state $\sigma^*$ such that
\begin{eqnarray*}
\rho\geq (1-C_w(\rho))\sigma^*.
\end{eqnarray*}
Apply the incoherent unitary $U$ on both sides, then we have
\begin{eqnarray*}
\rho'=U\rho U^\dag\geq (1-C_w(\rho))\sigma^* .
\end{eqnarray*}
Let us take the maximally coherent state $\ket{\psi_+}=\frac{1}{\sqrt{d}}\ket{j}$, then
\begin{eqnarray*}
\Innerm{\psi_+}{\rho'}{\psi_+}
\geq (1-C_w(\rho))\Innerm{\psi_+}{\sigma^*}{\psi_+},
\end{eqnarray*}
where
\begin{eqnarray*}
\Innerm{\psi_+}{\rho'}{\psi_+}=\frac{1}{d}\left(1-\sum_{i\neq j}|\rho_{ij}|\right)=\frac{1}{d}(1-C_{l_1}(\rho)),
\end{eqnarray*}
and $\Innerm{\psi_+}{\sigma^*}{\psi_+}=\frac{1}{d}$. That is,
\begin{eqnarray*}
\frac{1}{d}(1-C_{l_1}(\rho))\geq \frac{1}{d}(1-C_w(\rho)),
\end{eqnarray*}
which implies that $C_w(\rho)\geq C_{l_1}(\rho)$.
\end{proof}

~\\

\subsection{Exact coherence weight for generalized \texorpdfstring{$\bm{X}$}{}-states, Werner states, and Gisin states}
\label{subsec:exact_coh}
Here we find the exact analytical expressions of the coherence weight for some relevant classes of mixed states of $d$-dimensional single and bipartite quantum systems.

\subsubsection{Generalized \texorpdfstring{$X$}{}-states}
Generalized $X$ states \cite{genXstate2009, genXstate2010} form a $\left(2^{N+1}-1\right)$-parameter family of $N$-qubit states that encompass several classes of states like Werner states, Bell-diagonal states, and Dicke states. Given the complete characterization of the algebraic structure underlying the generalized $X$ states \cite{genXstate2009, genXstate2010}, these states are of paramount interest for analytical calculations in quantum information theory (see, e.g., Refs. \cite{Yu2007, Ali2010}). Proposition \ref{prop:6} holds for the generalized $X$ states, which have the form \cite{Piani2016}
\begin{equation*}
\rho=\left\{
\begin{array}{ccc}
\sum^{d/2}_{k=0}\eta_k,& d~\text{is even},\\
\sum^{\lfloor d/2\rfloor}_{k=0}\eta_k+\eta_c,& d ~\text{is odd},
\end{array}
\right.
\end{equation*}
where
\begin{eqnarray*}
\eta_k&=&\rho_{kk}\proj{k}+\rho_{k,d-1-k}\ket{k}\bra{d-1-k}\\
&+&\rho_{d-1-k,k}\ket{d-1-k}\bra{k}\\
&+&\rho_{d-1-k,d-1-k}\proj{d-1-k}
\end{eqnarray*}
and $\eta_c=\rho_{\lfloor d/2\rfloor+1}\proj{\lfloor d/2\rfloor+1}$. Thus, one can see that the $l_1$ norm of coherence of generalized $X$ states is always less than 1, no matter how large the dimension $d$ is.

\begin{widetext}

\begin{figure}
\centering     
\subfigure[]{\includegraphics[scale=0.25]{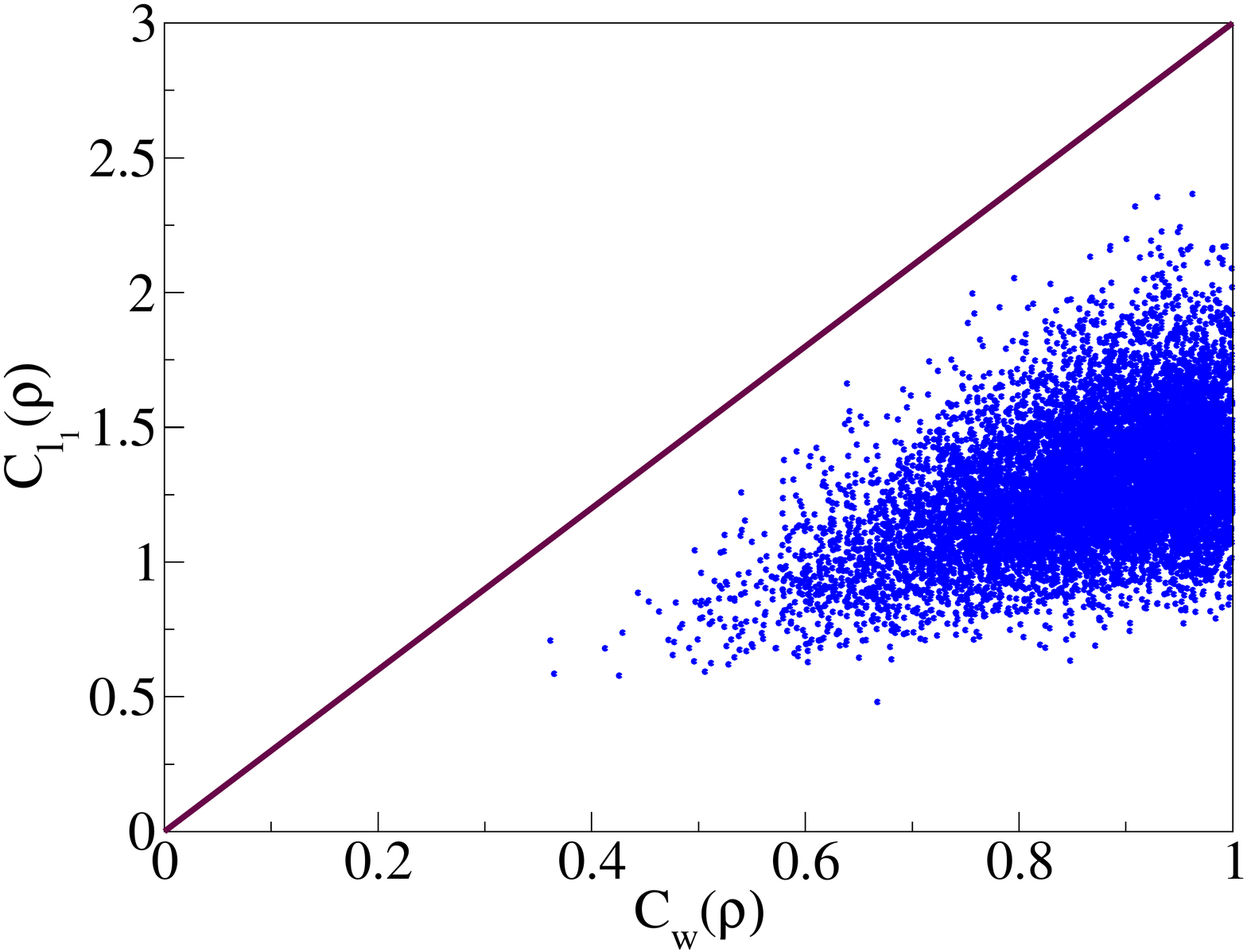}}
\subfigure[]{\includegraphics[scale=0.25]{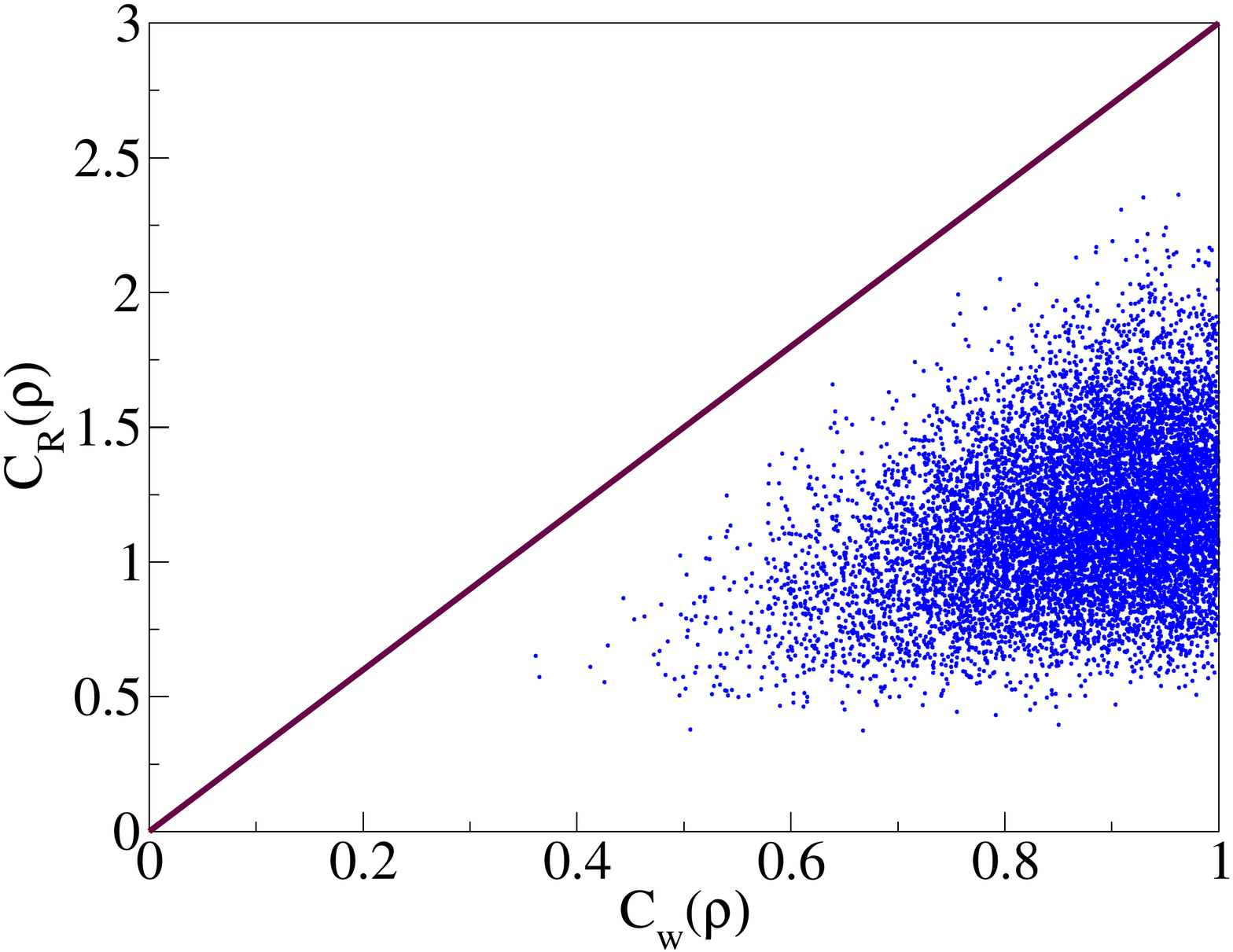}}
\caption{(Color online) Comparison between the coherence weight $C_w$ (horizontal axis) and (a) the $l_1$ norm of coherence $C_{l_1}$ and (b) the robustness of coherence  $C_{R}$ (vertical axes), for ${10}^4$ randomly generated quantum states (uniformly distributed with respect to the Haar measure) of dimension $4$. All the axes are unitless.}
\label{fig:cwvscl1vscrdim4}
\end{figure}

\end{widetext}
 
\subsubsection{Werner states}
A Werner state $\rho_W$ \cite{Werner1989, Werner2001}, was originally defined as a bipartite quantum state that is invariant under all unitary operators of the form $U\otimes U$. That is, a $d$-dimensional Werner state satisfies $\rho_W =(U\otimes U)\rho_W (U^{\dagger }\otimes U^{\dagger })$ for all unitary operators $U$ acting on the $d$-dimensional Hilbert space. Every Werner state can be written as a mixture of projectors onto the symmetric and antisymmetric subspaces. The only parameter that defines Werner states is the relative weight, say $\alpha\in[0,1]$, between the symmetric and antisymmetric subspaces.
Here, we show that for Werner states in any dimension $d$, the coherence weight, the robustness of coherence, and the $l_1$ norm coherence are all equal. 

\begin{prop}
 For $d$-dimensional Werner states $\rho_W(\alpha)=\alpha\frac{\mathbb{I}-F}{d(d-1)}+(1-\alpha)\frac{\mathbb{I}}{d^2}$
 with $F=\sum^{d-1}_{i,j=0}\ket{ij}\bra{ji}$ and $\alpha\in[0,1]$, we have
 \begin{eqnarray*}
 C_w(\rho_W(\alpha))= C_R(\rho_W(\alpha))= C_{l_1}(\rho_W(\alpha))=\alpha.
 \end{eqnarray*}
\end{prop}
\begin{proof}
First, it is easy to see that the $l_1$ norm of coherence of $\rho_w(\alpha)$ is $\alpha$. Now, since Werner states are generalized $X$ states, therefore $C_R(\rho_W(\alpha))= C_{l_1}(\rho_W(\alpha))$ as the robustness of coherence and the $l_1$ norm coherence are equal for generalized $X$ states \cite{Piani2016}. According to Proposition \ref{prop:6}, we have
$C_w(\rho_W(\alpha))\geq C_{l_1}(\rho_W(\alpha))=\alpha$ as $\rho_W(\alpha)$ are generalized $X$ states. Moreover, due to the convexity of the coherence weight, we have $C_w(\rho_W(\alpha))\leq \alpha C_w\left(\frac{\mathbb{I}-F}{d(d-1)}\right)+(1-\alpha)C_w\left(\frac{\mathbb{I}}{d^2}\right)=\alpha$. Therefore, $C_w(\rho_W(\alpha))= C_R(\rho_W(\alpha))= C_{l_1}(\rho_W(\alpha))=\alpha$.
\end{proof}

\subsubsection{Gisin states}
Gisin states are a family of two-qubit states introduced in Ref. \cite{Gisin1996} and can be written as $\rho_{\lambda,\theta}=\lambda\proj{\psi(\theta)}+(1-\lambda)\sigma_0$, where
$\ket{\psi(\theta)}=\sin\theta \ket{01}-\cos\theta \ket{10}$,
$\sigma_0=\frac{1}{2}\proj{00}+\frac{1}{2}\proj{11}$, $\lambda\in[0,1]$ and $\theta\in[0, 2\pi]$. These states are considered ``local'' in the sense that they do not violate any Bell-Clauser-Horne-Shimony-Holt inequality, but it was shown that after interaction with two independent environments, these states can violate a Bell inequality \cite{Gisin1996}.\\

Note that
$\rho_{\lambda, \theta}$ reduces to an incoherent state when $\lambda=0$ or $\sin\theta \cos\theta=0$. Thus, we consider the
nontrivial case: $\lambda\in(0, 1]$ and $\theta\in (0, \pi/2)$.
\begin{prop}
For  Gisin state $\rho_{\lambda,\theta}$ with $\lambda\in(0,1]$ and $\theta\in (0, \pi/2)$,
we have $C_w(\rho_{\lambda,\theta})=\lambda$ and $C_{l_1}(\rho_{\lambda,\theta})=C_{R}(\rho_{\lambda,\theta})=\lambda|\sin2\theta|$.
\end{prop}
\begin{proof}
It is easy to see that $C_{l_1}(\rho_{\lambda,\theta})=\lambda|\sin2\theta|$. Also, since Gisin states can be regarded as $X$ states, we have $\lambda|\sin2\theta|=C_{l_1}(\rho_{\lambda,\theta})=C_{R}(\rho_{\lambda,\theta})$ \cite{Piani2016}. Moreover, the convexity of $C_w$ implies that $C_w(\rho_{\lambda,\theta})\leq \lambda C_w(\proj{\psi(\theta)})+(1-\lambda)C_w(\sigma_0)\leq \lambda$. If $C_w(\rho_{\lambda,\theta})<\lambda$, then there exists an incoherent state $\rho_0$ such that $\rho_{\lambda,\theta}\geq (1-C_w(\rho_{\lambda,\theta}))\rho_0$. Now, consider
a pure state $\ket{\psi^{\perp}(\theta)}=\cos\theta \ket{01}+\sin\theta \ket{10}$ which is orthogonal to $\ket{\psi(\theta)}$. Then,
\begin{eqnarray*}
\Innerm{\psi^{\perp}(\theta)}{\rho_{\lambda,\theta}}{\psi^{\perp}(\theta)}
\geq(1-C_w(\rho_{\lambda,\theta}))\Innerm{\psi^{\perp}(\theta)}{\rho_0}{\psi^{\perp}(\theta)},
\end{eqnarray*}
implies that $0\geq \Innerm{\psi^{\perp}(\theta)}{\rho_0}{\psi^{\perp}(\theta)} =\cos^2\theta \Innerm{01}{\rho_0}{01}+ \sin^2\theta \Innerm{10}{\rho_0}{10}$.
Therefore, $\Innerm{01}{\rho_0}{01}=\Innerm{01}{\rho_0}{01}=0$ and we can write $\rho_0=p_{0}\proj{00}+p_{1}\proj{11} $ with $p_{0}+p_{1}=1$. Without loss of any generality, we can assume that $p_0\geq\frac{1}{2}$.
Now, $0=\lambda\iinner{00}{\psi(\theta)}\iinner{\psi(\theta)}{00}\geq (1-C_w(\rho_{\lambda,\theta}))\Innerm{00}{\rho_0}{00}-(1-\lambda)\Innerm{00}{\sigma_0}{00}$, which implies that $0\geq p_0(1-C_w(\rho_{\lambda,\theta}))-\frac{1}{2}(1-\lambda)>0$. Therefore, $C_w(\rho_{\lambda,\theta})=\lambda$ for Gisin states $\rho_{\lambda, \theta}$. This concludes the proof of the Proposition.
\end{proof}

\subsection{All pure coherent states have coherence weight 1}
If $\ket{\psi}$ is a coherent  pure state, then considering a decomposition of $\proj{\psi}$ as $\proj{\psi}=(1-s)\sigma+s\tau$, where $\sigma\in \mathcal{I}$ and $\tau\in\cD(\cH)$, implies that $s=1$. That is, the coherence weight of any pure coherent state is always equal to 1. Operationally this means that the coherence weight is a coarse-grained measure and cannot distinguish between different pure coherent states.\\

At this point, it is important to note that there exist some quantum states $\rho$ for which
$C_w(\rho)>C_{l_1}(\rho)$ and there also exist some states $\sigma$ such that $C_w(\sigma)\leq C_{l_1}(\sigma)$, based on the results obtained in this paper. This is different from the relationship $C_R(\rho)\leq C_{l_1}(\rho)$, which is true for any quantum state.

\subsection{Coherence weight for more general bipartite quantum states}
\begin{prop}
For any two quantum states $\rho_1\in\cD(\cH_1)$ and $\rho_2\in\cD(\cH_2)$, we have
\begin{align}
\label{ineq:subad}
&C_w(\rho_1\ot \rho_2)\leq C_w(\rho_1)+C_w(\rho_2)-C_w(\rho_1)C_w(\rho_2)~\mathrm{and}\\
&C_R(\rho_1\ot \rho_2)\leq C_R(\rho_1)+C_R(\rho_2)+C_R(\rho_1)C_R(\rho_2).
\end{align}
\end{prop}
\begin{proof}
Consider the optimal decompositions of $\rho_1$ and $\rho_2$ as
\begin{eqnarray*}
\rho_1=(1-C_w(\rho_1))\sigma^*_1+C_w(\rho_1)\tau^*_1,\\
\rho_2=(1-C_w(\rho_2))\sigma^*_2+C_w(\rho_2)\tau^*_2,
\end{eqnarray*}
where $\sigma^*_1\in\mathcal{I}_1$, $\sigma^*_2\in\mathcal{I}_2$, $\tau^*_1\in\cD(\cH_1)$, and $\tau^*_2\in\cD(\cH_2)$. Then, we have
\begin{eqnarray*}
\rho_1\ot \rho_2
&=&(1-C_w(\rho_1))(1-C_w(\rho_2))\sigma^*_1\ot\sigma^*_2\\
&&+(1-C_w(\rho_1))C_w(\rho_2)\sigma^*_1\ot \tau^*_2\\
&&+C_w(\rho_1)(1-C_w(\rho_2))\tau^*_1\ot\sigma^*_2\\
&&+C_w(\rho_1)C_w(\rho_2)\tau^*_1\ot\tau^*_2,
\end{eqnarray*}
where $\sigma^*_1\ot\sigma^*_2$ is an incoherent state in $\cD(\cH_1\ot\cH_2)$. The above equation implies that $C_w(\rho_1\ot \rho_2)\leq 1-(1-C_w(\rho_1))(1-C_w(\rho_2))=C_w(\rho_1)+C_w(\rho_2)-C_w(\rho_1)C_w(\rho_2)$.

Similarly, for the optimal incoherent states $\delta^*_1$ and $\delta^*_2$
such that
\begin{eqnarray*}
\frac{\rho_1+C_R(\rho_1)\delta^*_1}{1+C_R(\rho_1)}\in \mathcal{I}_1 ~\mathrm{and}~
\frac{\rho_2+C_R(\rho_2)\delta^*_2}{1+C_R(\rho_2)}\in \mathcal{I}_2,
\end{eqnarray*}
the state
\begin{eqnarray*}
\frac{[\rho_1+C_R(\rho_1)\delta^*_1]\ot [\rho_2+C_R(\rho_2)\delta^*_2]}{(1+C_R(\rho_1))(1+C_R(\rho_2))}
\end{eqnarray*}
is an incoherent state in $\cD(\cH_1\ot\cH_2)$. This implies that $C_R(\rho_1\ot \rho_2)\leq (C_R(\rho_1)+1) (C_R(\rho_2)+1)-1$. This completes the proof of the Proposition.
\end{proof}

\begin{widetext}

\begin{figure}
\centering     
\subfigure[]{\includegraphics[scale=0.25]{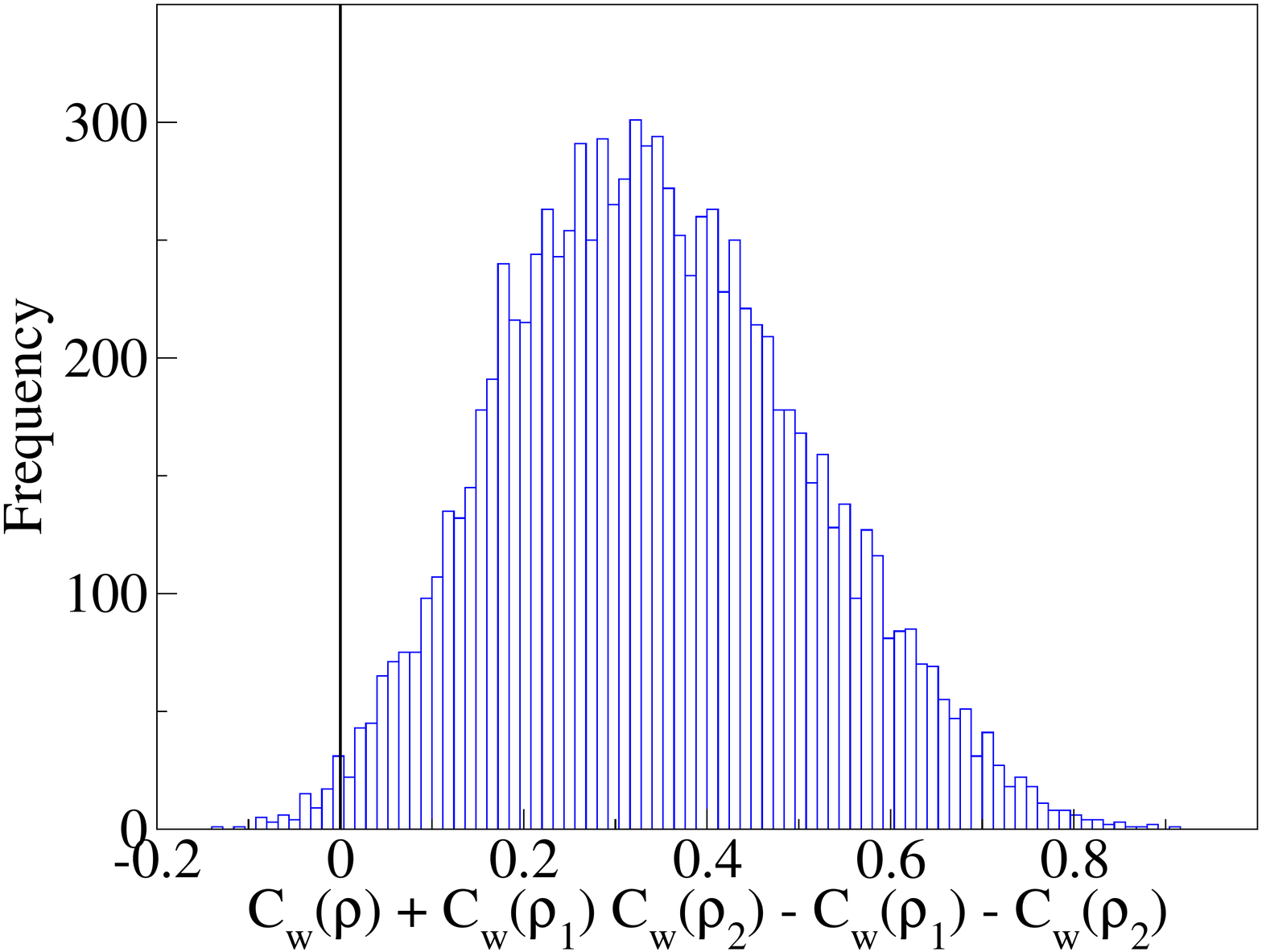}}
\subfigure[]{\includegraphics[scale=0.25]{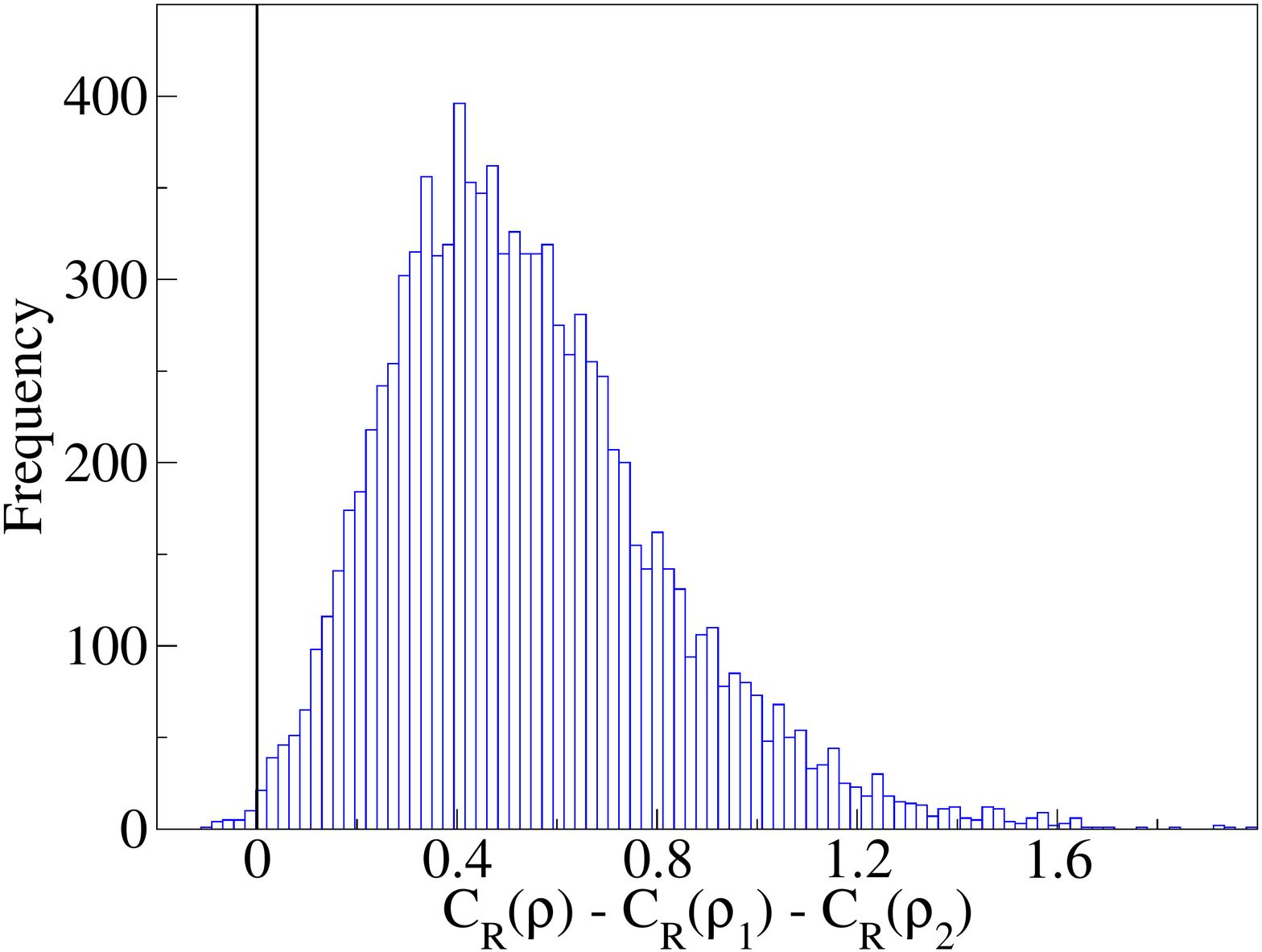}}
\caption{(Color online) Violation of Eqs. \eqref{w_dis} and \eqref{R_dis} for two-qubit mixed states. Figures (a) and (b) show the histograms of $C_w(\rho)+C_w(\rho_1)C_w(\rho_2) -  (C_w(\rho_1)+C_w(\rho_2))$ and $C_R(\rho) - (C_R(\rho_1)+C_R( \rho_2))$, respectively. All the axes are unitless. If Eqs. \eqref{w_dis} and \eqref{R_dis} were to hold true for arbitrary bipartite mixed states then there would not exist any states to the left of the vertical black line drawn in the figures. This, however, is not the case as one can see for ${10}^4$ randomly generated two-qubit mixed states (obtained by the partial tracing of the Haar distributed random pure states of dimension $4\ot4$). The coherence weight is calculated using a MATLAB code which we provide in the supplemental material \cite{supplement} and the robustness of coherence $C_R$ is calculated using QETLAB \cite{qetlab}.}
\label{fig:propositiontest}
\end{figure}

\end{widetext}

The inequality \eqref{ineq:subad} also implies
$C_w(\rho_1\ot \rho_2)\leq C_w(\rho_1)+C_w(\rho_2)$. Additionally, the relationship between the coherence weight (or robustness of coherence) in
$\rho$ and the coherence weight (or robustness of coherence) in $\rho_1$ and $\rho_2$ can play an important role in the distribution of coherence on bipartite systems \cite{Radhakrishnan2016}.

\begin{prop}
If $\rho \in \cD(\cH_1\ot\cH_2)$ is a pure state with reduced states
$\rho_1=\Ptr{2}{\rho}$ and $\rho_2=\Ptr{1}{\rho}$, then
\begin{eqnarray}
\label{w_dis}C_w(\rho)&&\geq C_w(\rho_1)+C_w(\rho_2)-C_w(\rho_1)C_w(\rho_2),\\
\label{R_dis}C_R(\rho)&&\geq C_R(\rho_1)+C_R( \rho_2).
\end{eqnarray}
However, for general bipartite states, the relationships \eqref{w_dis} and \eqref{R_dis}
need not hold (see Figure \ref{fig:propositiontest}).
\end{prop}

\begin{proof}
For coherence weight, the result comes directly from the fact that the coherence weight of any coherent pure state always attains the maximal value 1.

If $\rho$ is a bipartite pure state with reduced states $\rho_1=\Ptr{2}{\rho}$ and $\rho_2=\Ptr{1}{\rho}$, then
\begin{eqnarray}\nonumber
C_R(\rho)=C_{l_1}(\rho)
&\geq& C_{l_1}(\rho_1)+C_{l_1}(\rho_2)\\
&\geq& C_R(\rho_1)+C_R(\rho_2).
\end{eqnarray}
Here the first inequality follows from the
fact that the $l_1$ norm of coherence in any bipartite state is larger than the sum of the $l_1$ norm of coherence in the reduced states \cite{Bu2015a}. The second inequality comes from the relation $C_{l_1}(\rho)\geq C_R(\rho)$ \cite{Piani2016}.
\end{proof}

\section{Conclusion}\label{sec:con}
In this paper, we introduce the notion of the asymmetry weight and the coherence weight in the resource theories of asymmetry and coherence, respectively. The asymmetry and the coherence weight satisfy some interesting properties such as convexity and monotonicity, and thereby qualify as bona fide measures of asymmetry and coherence, respectively. These measures can also be interpreted operationally as the minimum number of genuine resource states needed in the preparation process of a given quantum state under the restrictions imposed by the relevant resource theory. Interestingly, these measures can easily be computed numerically for arbitrary quantum states since they can be characterized as the solutions of the corresponding semidefinite programs. Moreover,   
we show that coherent (asymmetric) pure quantum states have coherence (asymmetry) weight equal to 1. Importantly, we analytically find the exact coherence weight for some classes of bipartite mixed states such as the Werner states and Gisin states, which are subsets of the generalized $X$ states. For Werner states, we find that the coherence weight, the robustness of coherence, and the $l_1$ norm of coherence are all equal and are given by a single letter formula. Similarly, for Gisin states, we find closed-form expressions for the coherence weight, the robustness of coherence, and the $l_1$ norm of coherence.
In general, for bipartite states, we establish various relationships between the coherence weight, the robustness of coherence, and the $l_1$ norm of coherence. Our results imply that there exist some quantum states for which the coherence weight is greater than or equal to the $l_1$ norm of coherence and there also exist some states for which the opposite holds. This is in stark contrast to the fact that the $l_1$ norm of coherence is always greater than or equal to the robustness of coherence. 

Moreover, the SDP form of the asymmetry weight readily allows us to establish a plausible connection with the (state-dependent) asymmetry witnesses. As the swap entanglement witness can be viewed as a special asymmetry witness, this suggests that asymmetry may be applied to detect the existence of entanglement in a given bipartite state. Furthermore, in the context of the distribution of quantum coherence, we provide some relationships between the coherence weight (the robustness of coherence) of a given bipartite state and the coherence weight (the robustness of coherence) of its marginals.

We hope that the operational interpretation and the ease to calculate the coherence (asymmetry) weight for an arbitrary quantum state make these measures desirable and therefore, may help in improving our understanding of these two resources at a quantitative level. Also, given the connection between the asymmetry weight and the entanglement witnesses, it will be an important future direction to find the exact relationship between the asymmetry and the entanglement.\\

\begin{acknowledgments}
This paper is supported by the National Natural Science Foundation of China (Grants No. 11171301 and No. 11571307 ). K.B. acknowledges L. Li for various discussions and help. N.A. acknowledges the Harish-Chandra Research Institute for its hospitality during the preparation of this paper. U.S. acknowledges support from a research fellowship of Department of Atomic Energy, Government of India.
\end{acknowledgments}

\bibliographystyle{apsrev4-1}

 \bibliography{Asym-lit}

\end{document}